\documentclass[runningheads]{maf2008}

\usepackage[]{latexsym}
\usepackage{amsmath,amssymb,array,enumerate,graphicx}

\input{psfig.sty}

\begin{document}

\pagestyle{myheadings} \setcounter{page}{1}

\mainmatter

\title{Feasibility of Portfolio Optimization under Coherent Risk Measures} %\thanks{Supported by$\ldots$}}

\titlerunning{Feasibility of Portfolio Optimization $\ldots$}

\author{Imre Kondor\inst{1}\inst{2}\inst{4}
\and Istv\'an Varga-Haszonits\inst{2}\inst{3}\inst{5}}

\authorrunning{I. Kondor and I. Varga-Haszonits}

\institute{Collegium Budapest -- Institute for Advanced Studies,\\
Szenth\'aroms\'ag u. 2, H-1014 Budapest, Hungary \and
Department of Physics of Complex Systems,\\ E\"otv\"os University,\\ P\'azm\'any P\'eter s\'et\'any 1/A, H-1117 Budapest, Hungary\and
Analytics Department of Fixed Income Division,\\Morgan Stanley Hungary Analytics,\\ De\'ak Ferenc u. 15, H-1052 Budapest, Hungary\and
\email{kondor@colbud.hu}\and
\email{Istvan.Varga-Haszonits@morganstanley.com}}

\maketitle

\begin{abstract}
It is shown that the axioms for coherent risk measures imply that
whenever there is an asset in a portfolio that dominates the others in a
given sample (which happens with finite probability even for large
samples), then this portfolio cannot be optimized under any coherent
measure on that sample, and the risk measure diverges to minus infinity.
This instability was first discovered on the special example of Expected
Shortfall which is used here both as an illustration and as a prompt for
generalization.
\smallskip


\smallskip

\noindent {\bf J.E.L. classification.} \textsc{\small G11, C13, D81.}
\end{abstract}

\bigskip

\section{Introduction}

The proper mathematical characterization of risk is of central importance in finance: In the introductory chapter of a recent volume entirely dedicated to risk measures, Giorgio Szeg\"o called the upsurge of interest in risk measures from 1997 onwards the third major revolution in finance theory \cite{szego}.  The choice of the adequate risk measure is a highly nontrivial task that involves deep, in most cases implicit, considerations concerning the attitudes of market players, the structure of markets, the instruments in question, etc. No wonder that there is no universally accepted risk measure, the ones in use range from rather ad hoc characteristics to axiomatic constructs. Recently, value at risk (VaR) \cite{jorion,riskmetrics} has gained widespread use, in practice as well as in regulation. VaR has, however, been criticized by a number of academics, because as a quantile it has no reason to be convex, and indeed it is easy to construct portfolios for which VaR seriously violates convexity \cite{artzner2,acerbi1,embrechts}. The shortcomings of VaR and other considerations led Artzner et al. \cite{artzner1,artzner2} to the introduction of coherent risk measures. Independently and essentially simultaneously similar considerations were put forward by a number of researchers \cite{wang,hodges,carr} in various contexts. The seminal paper of Artzner et al. triggered a rich branch of research with the number of papers inspired by it ranging in the hundreds. Expected Shortfall (ES), a particularly appealing risk measure, was quickly demonstrated to be coherent \cite{acerbi2,acerbi3}, and Acerbi \cite{acerbi4} went on to construct a whole class of coherent risk measures called spectral measures of which ES was a special case. ES is essentially the conditional average loss above a high threshold. As such, it is easily related to VaR that can be identified with this threshold (although for ES this cutoff is defined in probability, not in money). This explains why ES is often identified with conditional value at risk CVaR, although for discrete distributions there is a subtle but significant difference between the two (see \cite{acerbi3} for a careful analysis of the relationship between them). As a conditional average, ES is also a more plausible measure than VaR, it is easy to determine from empirical distributions, and the optimization under ES can be reduced to a linear programming task \cite{rockafellar}. All these desirable features have made ES very popular among academics.

Therefore it came as a complete surprise for us when we recognized a few years ago that the optimization of some portfolios, especially the large ones, was not always feasible under ES: in some of the samples the risk measure became unbounded from below \cite{kondor1}. This is a probabilistic phenomenon: it depends on the sample. It is natural to ask whether this is a special pathology of ES or other risk measures also display this instability. The purpose of the present paper is to answer this question.

The rest of the paper is organized as follows. In Section \ref{sec:es} we recapitulate some of the results related to the instability of ES. We also display the simple financial content of this mathematical phenomenon which will immediately make it obvious that the instability must be a common property of every coherent measure. A formal proof of this is provided in Section \ref{sec:crm}. The paper ends on a short Summary.

\section{The Instability of Expected Shortfall}
\label{sec:es}

Let us begin with a special case of ES corresponding to the limit where the cutoff $\alpha$ beyond which the conditional average is calculated goes to unity. Then only the single worst loss will contribute and the optimization over the portfolio weights under this risk measure consists in finding the best linear combination of the worst losses. This risk measure is sometimes called the Maximal Loss  (ML) or Minimax measure \cite{young}. For this particular measure, and assuming that the returns follow an elliptical distribution, it proved possible to exactly determine the probability of having a solution as function of the portfolio size $N$ and the sample size $T$ \cite{kondor1} with the following result:
\begin{equation}
p(N,T)=\Theta(T-N)\cdot\frac{1}{2^{T-1}}\sum_{k=N-1}^{T-1}C^k_{T-1}\label{eq:minimaxprob}
\end{equation}
where $\Theta(x)$ is the step function (equal to 1 if $x\ge1$ and equal to 0, if $x<0$), and the symbol $C^k_n$ represents the binomial coefficients.

As shown by \eqref{eq:minimaxprob}, for short time series $T<N$ the optimization problem never has a solution, but this is rather natural: if we do not have enough data it is impossible to find an optimum. Other, more conventional risk measures (variance and mean absolute deviation or MAD \cite{konno}) also fail to allow for a solution in that parameter range. The remarkable fact about \eqref{eq:minimaxprob} is that for finite $N$ and $T$ the probability of having a solution is less than unity even for $T>N$.  In the limit $N$, $T$ going to infinity with their ratio fixed, the probability distribution becomes a step function: for $N/T<1/2$ it goes to 1, for $N/T>1/2$ to zero.

For the case of ES with a cutoff parameter $\alpha$ and again in the limit $N$, $T$ going to infinity, we find that the critical value of the ratio $N/T$ separating the region where the problem always has a solution from the one where it does not, varies with $\alpha$. This way, we get a kind of phase diagram separating the region where the optimization problem is feasible from that where it is not. This curve is illustrated in Fig. \ref{fig:es}. As can be seen, the curve is sloping down from right to left, so the critical value of $N/T$ where the optimization task becomes unfeasible is decreasing with the cutoff $\alpha$. The phase diagram has been numerically determined in \cite{kondor1}. Subsequently it even proved to be possible to exactly calculate this phase diagram \cite{ciliberti} via the method of replicas, borrowed from the theory of random systems. The calculation shows that the curve behaves in a non-analytic fashion near $\alpha=1$, the drop from $1/2$ being exponentially small for $\alpha$'s in the vicinity of 1. Since the $\alpha$'s of practical interest are fairly close to 1 (such as $.95$ or $.99$) we can say that the critical value of $N/T$ is very close to $1/2$ for any practically interesting case.
\begin{figure}
\includegraphics[width=10cm]{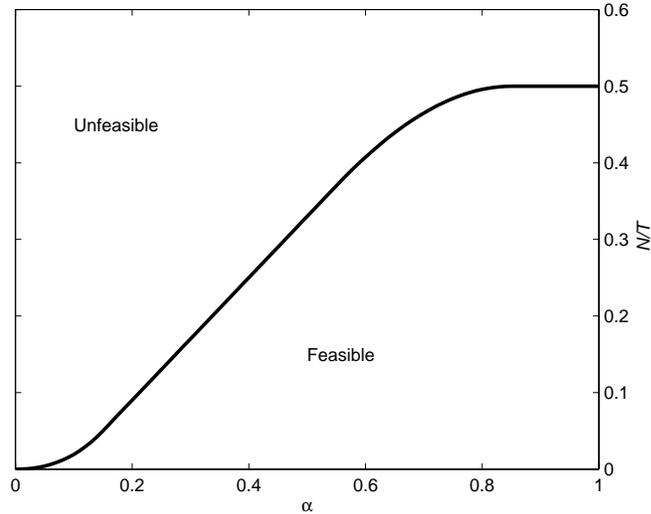}
\caption{\label{fig:es} The phase diagram of Expected Shortfall. In the region called 'Unfeasible' in the figure the optimization under ES does not have a solution. The transition across the line separating the feasible and unfeasible regions is sharp if $N$,$T$ go to infinity with their ratio kept fixed.}
\end{figure}

To put this phenomenon in perspective it is worth making a comparison with the case of the variance. As long as $T>N$ the covariance matrix is always positive definite, so the optimization can always be carried out. It is well known, however, that the rank of the covariance matrix is the smaller of $N$ and $T$, so when $T$ goes to $N$ from above the covariance matrix develops a zero eigenvalue, and the curvature of the risk measure becomes zero in a certain direction. With this the optimization problem loses its meaning. As $T$ decreases further, more and more such neutral directions appear, but the variance as a risk measure never becomes unbounded. The behavior of ES as a risk measure is both similar to and different from this. For finite $N<T$ there are always samples for which the optimization does not have a solution, but the reason is not the appearance of neutral directions, instead the risk measure ES becomes unbounded from below. For $N/T$ small the probability of such samples is small. As $N/T$ approaches its critical value, the probability of samples for which the optimization is unfeasible is growing and as we go above the critical $N/T$ it fast approaches 1. Finally, when $N/T$ goes above 1 all the samples become unfeasible. For $N,T\to\infty$ with their ratio fixed, this transition becomes sharp: below the critical value of $N/T$ all the samples are feasible, above none of them.

Note that the critical value of $N/T$ is always smaller for ES than for either the variance or MAD, so ES is more sensitive to sample to sample fluctuations (noise) than either the variance or MAD \cite{kondor1}.  The phase diagram does not only separate the regions where the optimization under ES can or cannot be carried out, but it is also the locus where the estimation error diverges. All these features and several more details were thoroughly discussed in \cite{kondor1}, but the financial content of this instability was left somewhat implicit.

As a matter of fact, the message is almost self-evident. Let us just consider the simplest case imaginable: $N=T=2$ for the Minimax problem. Let $x_{it}$ be the returns on securities $i=1,2$ at times $t=1,2$, and let $w_i$ ($i=1,2$) be the portfolio weights. Then the task is to find the weights $w_i$ that make the maximal loss
\begin{equation}
\max_{t}\left(-\sum_{i=1}^Nw_ix_{it}\right)\label{eq:maxloss}
\end{equation}
minimal. As the weights sum to 1 we let $w_1=w$, and $w_2=1-w$. Then the losses of the portfolio over periods $t=1,2$ are the following:
\begin{gather*}
y_1=-wx_{11}-(1-w)x_{21}=w(x_{21}-x_{11})-x_{21},\\
y_2=-wx_{12}-(1-w)x_{22}=w(x_{22}-x_{12})-x_{22}.
\end{gather*}
Whether $\max\{y_1,y_2\}$ is bounded from below or not, depends on the slopes of these two straight lines, see Fig. \ref{fig:mm}.
\begin{center}
\begin{figure}
\includegraphics[width=12cm]{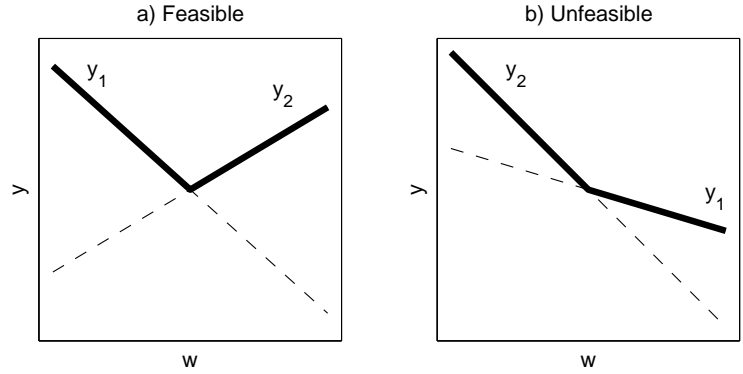}
\caption{\label{fig:mm}The losses $y_1$ and $y_2$ (dashed lines) and their maximum (heavy lines) as a function of the portfolio weight $w$. Panel a) shows the case where the maximal loss is bounded from below so we have a minimum, while in Panel b) the maximal loss is unbounded, hence there is no minimum.}
\end{figure}
\end{center}
If the two slopes are of opposite sign, $\max\{y_1,y_2\}$ is bounded from below, and the minimax problem has a solution. If the two slopes are of the same sign, there is no solution, the minimum runs away to infinity, and the value of the objective function in \eqref{eq:maxloss} goes to minus infinity. When do the two slopes have the same sign? This happens if
\begin{enumerate}[(a)]
\item\label{item:case_a} either $x_{21}>x_{11}$ {\it and} $x_{22}>x_{12}$,
\item\label{item:case_b} or $x_{21}<x_{11}$ {\it and} $x_{22}<x_{12}$.
\end{enumerate}
In case (\ref{item:case_a}) asset $i=2$ dominates asset $i=1$ in the sense that the return on asset $i=2$ is larger than the return on asset $i=1$ for all $t=1,2$, and the solution runs away to infinity such that $w_2\to+\infty$ and $w_1\to-\infty$. In the opposite case (\ref{item:case_b}) asset $i=1$ dominates and the solution runs away to infinity with $w_1\to+\infty$ and $w_2\to-\infty$. But this is obvious! If the return on one of the assets is always larger than that on the other, we have to go infinitely long in the dominating asset, and infinitely short in the other one, and this will produce an infinitely large negative risk, i.e. no risk at all.

The generalization to ES and beyond is immediate: If we have several assets in the portfolio and we have larger, but finite samples (i.e. finite time series segments), it may always happen that one asset, or a combination of assets, dominate the rest {\it in the given sample}, even if they do not dominate them in general. In such a case the optimization of portfolios under a risk measure that penalize large losses more severely than smaller losses will inevitably lead to the result that we have to go infinitely long in the dominating asset(s) and infinitely short in the rest, thereby ending up with the illusion of a portfolio without risk. Of course, all this strongly depends on the size of the portfolio and on the size of the sample. For $N$ increasing it is easier and easier to find a dominating combination, while for an increasing sample $T$ it is less and less probable that an atypical fluctuation may lead us to a false conclusion. This is evidently born out by the features displayed in Fig. \ref{fig:es}: for a given $\alpha$ the feasibility of the task is lost when $N/T$ exceeds a certain threshold, while for a given $N/T$ the feasibility is lost with decreasing $\alpha$, as we are allowing, in addition to the worst case, a larger and larger part of the worst cases to enter into competition.

Runaway solutions may, of course, arise only if we allow arbitrarily large short positions. The main message does not change much, however, even if we impose a ban on short selling, or any other constraint on the weights. For finite $T$ there will always be a chance for one or the other security to dominate, and then the optimization will lead us to the conclusion that we should go long in the dominating item to the maximum extent allowed by the constraints. This will not result in infinite positions, nevertheless the solution will still lack any stability: in the next sample another security may dominate, so the solution will jump around on the 'walls' of the allowed region. A related instability has been studied in the case of the variance in \cite{gulyas}.
The insight gained on the above example leads us to the conjecture that coherent measures may, in general, be sensitive to this type of instability for finite samples. This will, indeed, turn out to be the case, as shown by the formal argument developed in the next Section.

\section{The instability of coherent measures}
\label{sec:crm}

In addition to the notation we have already introduced, that is $N$ for the number of assets in the portfolio, $T$ for the number of observation periods, $x_{it}$ the return on asset $i$ over the observation period $t$, and $w_i$ the portfolio weight of asset $i$, we are also going to use the notation $\mathbf{w}$ for the $N$ dimensional vector of portfolio weights and $\mathbf{X}$ for the $N\times T$ matrix of the observed returns (i.e. the matrix with entries $x_{it}$ ).

In portfolio optimization it is customary to assume that the sum of all portfolio weights is unity. Hence, to simplify discussion we are going to use the following definition:
\begin{definition}{Normalized portfolio.} We say that a portfolio $\mathbf{w}$ is normalized, if
\begin{equation}
\sum_{i=1}^Nw_i=1.
\end{equation}
\end{definition}

\subsection{Dominant portfolios}

As explained in the previous section, the minimization of ES is not feasible, if there is a dominating combination of assets on the given sample. In order to spell out this assertion in a more formal and more general way, let us define exactly what we mean by {\it domination}.
\begin{definition}\label{def:dominance}
Let $\mathbf{u}$ and $\mathbf{v}$ be two portfolio vectors, moreover, let $\mathbf{u}'$ and $\mathbf{v}'$ be their normalized counterparts (i.e. they are normalized and parallel to $\mathbf{u}$ and $\mathbf{v}$ respectively). We say that $\mathbf{u}$ dominates $\mathbf{v}$ on the sample $\mathbf{X}$, if for all $t=1,2,...,T$ we have
\begin{equation}
\sum_{i=1}^Nu'_ix_{it}\ge\sum_{i=1}^Nv'_ix_{it},\label{eq:domin}
\end{equation}
The dominance of $\mathbf{u}$ over $\mathbf{v}$ is denoted by $\mathbf{u}\succeq\mathbf{v}$.

Furthermore, we say that $\mathbf{u}$ strictly dominates $\mathbf{v}$ on the sample $\mathbf{X}$ if $\mathbf{u}\succeq\mathbf{v}$ and in addition there is at least one observation $1\le t\le T$ such that
\begin{equation}
\sum_{i=1}^Nu'_ix_{it}>\sum_{i=1}^Nv'_ix_{it}.\label{eq:strictdomin}
\end{equation}
The strict dominance of $\mathbf{u}$ over $\mathbf{v}$ is denoted by $\mathbf{u}\succ\mathbf{v}$.
\end{definition}
Obviously, parallel portfolio vectors are equivalent, but in order to compare the performance of two portfolios we have to make sure that both of them are normalized --- this is why the dominance conditions \eqref{eq:domin} and \eqref{eq:strictdomin} should apply to $\mathbf{u}'$ and $\mathbf{v}'$ rather than $\mathbf{u}$ and $\mathbf{v}$. Definition \ref{def:dominance} simply states that a portfolio dominates another one on a given sample, if the return on the first is not smaller than that on the second for any observation. If it is higher for at least one (or more) observations then this dominance is strict.

It is important to emphasize that the dominance (or strict dominance) relation, as defined above, between two portfolios is specific to the underlying sample. Clearly, if two vectors form a dominant-dominated pair for (almost) all possible samples then we have an arbitrage opportunity. At the same time, it is obvious that dominance on a single sample does not imply arbitrage: even if $\mathbf{u}\succ\mathbf{v}$ on sample $\mathbf{X}_1$ (giving the false illusion of an arbitrage opportunity) it is quite possible that neither of them dominates the other on another sample $\mathbf{X}_2$.\footnote{In fact, the notation $\mathbf{u}\succ\mathbf{v}$ is somewhat sloppy and $\mathbf{u}\succ_{\mathbf{X}}\mathbf{v}$ would be more appropriate. However, we refrain from this rather awkward notation, as it will always be clear on what sample we investigate dominance.}

\subsection{Optimization of coherent risk measures}

Let us assume that we have a coherent risk measure $\rho$ and we want to optimize our portfolio under this measure, on the basis of data making up the sample $\mathbf{X}$. Since we do not know the true data generating process, we cannot tell exactly what the risk of a portfolio $\mathbf{w}$ is, but we can estimate it on the basis of the given sample. Let us denote this estimate by $\hat{\rho}\left(\mathbf{w}|\mathbf{X}\right)$. Since $\rho$ is coherent, we require that its estimator $\hat{\rho}$ also be coherent on the sample $\mathbf{X}$. Let us recall the axioms of coherence and adapt them to the present context:
\begin{definition}
We say that the risk measure estimator $\hat{\rho}$ is coherent on sample $\mathbf{X}$, if it fulfills all of the following requirements (where $\mathbf{0}$ denotes the null-vector):
\begin{enumerate}
\item {\it Monotonicity:} $\mathbf{u}\succeq\mathbf{0}\Rightarrow \hat{\rho}\left(\mathbf{u}|\mathbf{X}\right)\le0$
\item {\it Sub-additivity:} $\hat{\rho}(\mathbf{u}+\mathbf{v}|\mathbf{X})\le\hat{\rho}(\mathbf{u}|\mathbf{X})+\hat{\rho}(\mathbf{v}|\mathbf{X})$
\item {\it Positive homogeneity:} $a>0\Rightarrow\hat{\rho}(a\mathbf{u}|\mathbf{X})=a\hat{\rho}(\mathbf{u}|\mathbf{X})$
\item {\it Translational invariance:} $a\in\mathbb{R}\Rightarrow\hat{\rho}(\mathbf{u}|\mathbf{X}+a)=\hat{\rho}(\mathbf{u}|\mathbf{X})-a$
\end{enumerate}
where $\mathbf{X}+a$ denotes the matrix obtained by adding the number $\mathrm{a}$ to each element of $\mathbf{X}$.
\end{definition}
The global risk minimization problem for the finite sample $\mathbf{X}$ is
\begin{gather}
\min_{\mathbf{w}}\hat{\rho}(\mathbf{w}|\mathbf{X}),\label{eq:crmopt}\\
\sum_{i=1}^Nw_i=1.\nonumber
\end{gather}
The following theorem states that the existence of an optimal portfolio is not automatically guaranteed, and it is a property of the underlying sample.
\begin{theorem}\label{theorem:crm}
\label{theor:crmstab}
If there exist two portfolios $\mathbf{u}$ and $\mathbf{v}$ so that $\mathbf{u}\succ\mathbf{v}$ then \eqref{eq:crmopt} has no solution.
\end{theorem}
\begin{proof}
Let us use proof by contradiction and assume that $\mathbf{w}^{(0)}$ is the (normalized) minimal risk portfolio. Without loss of generality we can assume that $\mathbf{u}$ and $\mathbf{v}$ are normalized. Let $\mathbf{0}$ denote the null vector (empty portfolio), then clearly $\mathbf{u}-\mathbf{v}\succ\mathbf{0}$, hence, due to the monotonicity and translational invariance axioms, we get that
\begin{equation}
\hat{\rho}(\mathbf{u}-\mathbf{v}|\mathbf{X})<0.\label{eq:dominance}
\end{equation}
Let us now consider the following portfolio:
\begin{equation*}
\mathbf{w}^{(1)}=\mathbf{w}^{(0)}+a\left(\mathbf{u}-\mathbf{v}\right),
\end{equation*}
where $a$ is a positive real number. It is straightforward to see that $\mathbf{w}^{(1)}$ is also normalized. Using the sub-additivity and positive homogeneity axioms:
\begin{equation*}
\hat{\rho}\left(\mathbf{w}^{(1)}|\mathbf{X}\right)=\hat{\rho}\left(\mathbf{w}^{(0)}+a\left(\mathbf{u}-\mathbf{v}\right)|\mathbf{X}\right)
\le\hat{\rho}\left(\mathbf{w}^{(0)}|\mathbf{X}\right)+a\hat{\rho}\left(\mathbf{u}-\mathbf{v}|\mathbf{X}\right).
\end{equation*}
Because of \eqref{eq:dominance} and the positivity of $a$ the last term on the right hand side is negative. Hence we get:
\begin{equation*}
\hat{\rho}\left(\mathbf{w}^{(1)}|\mathbf{X}\right)<\hat{\rho}\left(\mathbf{w}^{(0)}|\mathbf{X}\right),
\end{equation*}
which obviously contradicts the assumption that $\mathbf{w}^{(0)}$ is optimal.
\end{proof}

This theorem states that the existence of a dominant-dominated portfolio pair is a sufficient condition for the non-existence of the optimum. This condition, however, is not necessary in general: the optimization problem \eqref{eq:crmopt} may be unfeasible even if there is no dominating portfolio.

One important exception is the minimization of the Maximal Loss measure (Minimax problem), which is formally:
\begin{gather}
\min_{\mathbf{w}}\max_{1\le t\le T}\left(-\sum_{i=1}^Nw_ix_{it}\right),\label{eq:minimax}\\
\sum_{i=1}^Nw_i=1,\nonumber
\end{gather}

\begin{theorem}\label{theorem:minimax}
The optimization problem \eqref{eq:minimax} has no solution, if and only if there exists a pair of portfolios so that one of them strictly dominates the other.
\end{theorem}
\begin{proof}
Due to Theorem \ref{theor:crmstab} we only have to prove that there must be a pair of dominating-dominated portfolios, when the Minimax problem has no solution. Let us assume the contrary: there is no portfolio that strictly dominates any other portfolios. It can be shown \cite{young} that the Minimax problem is equivalent to the following linear programming task:
\begin{gather}
\min u\nonumber\\
u\ge-\sum_{i=1}^Nw_ix_{it}\label{eq:minimax2}\\
\sum_{i=1}^Nw_i=1\nonumber
\end{gather}
Therefore, having no solution means that for any real number $u$ there is a normalized portfolio $\mathbf{w}$ such that $u>-\sum_{i=1}^Nw_ix_{it}$ for any $t=1,2,...,T$. (That, is the portfolio return is not bounded in the space of portfolios.)

More specifically, let us consider the maximum of all returns in our sample $x_{it}$ and let the indices of this (or one of these) maximal observation(s) be $i_1$ and $t_1$. (Since the number of assets and the number of observations are both finite, this maximum must exist, although it may not be unique.) Letting $u=-x_{i_1t_1}$ we have a portfolio $\mathbf{w}^{(1)}$ such that $x_{i_1t_1}<\sum_{i=1}^Nw_i^{(1)}x_{it}$ for any $t=1,2,...,T$ so eventually:
\begin{equation}
x_{i_1t_1}<\min_{1\le t\le T}\left(\sum_{i=1}^Nw_i^{(1)}x_{it}\right)\label{eq:i1t1}
\end{equation}
Since there is no dominating portfolio, for $\mathbf{w}^{(1)}$ there exist indices $1\le i_2\le N$ and $1\le t_2\le T$ such that $x_{i_2t_2}>\sum_{i=1}^Nw_i^{(1)}x_{it_2}$ (if that were not true $\mathbf{w}^{(1)}$ itself would dominate all of the $N$ assets, that is, any possible portfolio). Therefore:
\begin{equation}
x_{i_2t_2}>\min_{1\le t\le T}\left(\sum_{i=1}^Nw_i^{(1)}x_{it}\right).\label{eq:i2t2}
\end{equation}
\eqref{eq:i1t1} and \eqref{eq:i2t2} together mean that $x_{i_1t_1}<x_{i_2t_2}$. Since $x_{i_1t_1}$ is the maximum of all sample returns, this is obviously a contradiction.
\end{proof}

It is important to stress that that neither Theorem \ref{theorem:crm} nor Theorem \ref{theorem:minimax} assume anything about the probability distribution of the asset returns. Further results can immediately be inferred if we now assume that the returns follow an elliptical distribution.

Then equation \eqref{eq:minimaxprob} in Section \ref{sec:es} states the probability of the Minimax problem being feasible. Theorem \ref{theorem:minimax} asserts that the Minimax problem is unfeasible, if and only if there is a dominant-dominated pair of portfolios. Thus we arrive at the following corollaries (where $p(N,T)$ is the probability given by \eqref{eq:minimaxprob}):
\begin{corollary}
For returns drawn from an elliptical distribution the probability of the existence of a pair of portfolios such that one of them dominates the other on the sample $\mathbf{X}$ is $1-p(N,T)$.
\end{corollary}
\begin{corollary}\label{coroll:prob}
For elliptically distributed returns the probability of the unfeasibility of the portfolio optimization problem under any coherent measure on the sample $\mathbf{X}$ is at least $1-p(N,T)$.
\end{corollary}
Moreover, in Section \ref{sec:es} we saw that if $N,T$ go to infinity such that $N/T$ is held constant then for $N/T<1/2$ we have $p(N,T)\to1$, while for $N/T>1/2$ we have $p(N,T)\to0$. Due to Corollary \ref{coroll:prob} $1-p(N,T)$ is a lower bound of the probability of having no solution for an arbitrary coherent measure. Thus we can also state:
\begin{corollary}
If, similarly to the case of ES, there is a sharp feasibility--non-feasibility transition for an arbitrary coherent measure in the limit where $N,T$ go to infinity with their ratio fixed, then the critical value of $N/T$ must be smaller than 1/2 for portfolios of elliptically distributed returns.
\end{corollary}

Similar statements can be made for non-elliptically distributed returns, the only difference being that the probability of finding a solution for the Minimax problem is not known explicitly in the general case.

Let us now turn to the financial interpretation of the findings of this section. Essentially, Theorem \ref{theorem:crm} used an arbitrage argument to show that the risk is unbounded from below. Assuming that portfolios $\mathbf{u}$ and $\mathbf{v}$ are such that $\mathbf{u}\succ\mathbf{v}$ we can go $a$ unit long of $\mathbf{u}$ and $a$ unit short of $\mathbf{v}$ so that the portfolio returns on the given sample will be given by the vector $a\left(\mathbf{u}^{\mathrm{T}}\mathbf{X}-\mathbf{v}^{\mathrm{T}}\mathbf{X}\right)$, whose elements are strictly positive. Moreover, the higher the leverage, the bigger the profit, so there is no optimum: we can always do better by increasing the value of $a$.

There is general consensus among financial experts that on real markets such arbitrage opportunities rarely occur and quickly disappear. However, in finite length time series such a dominance may be observed with non-zero probability even between assets without a dominance relation between them on the long run. Clearly, the shorter the time series, the higher the probability of the occurrence of such a finite size dominance.

\section{Summary}
\label{sec:sum}

We have shown that the axioms of Artzner et al. \cite{artzner1,artzner2} imply that for finite $T$ there will always be samples for which the portfolio optimization cannot be carried out under a given coherent risk measure, because the measure becomes unbounded from below. We would like to stress again that the point here is not that such an instability can occur for large values of the ratio of the portfolio size $N$ to the sample size $T$: that is a general feature of any reasonable risk measure. The point is that the optimization under coherent measures can become impossible even for small $N/T$'s. This is a purely random effect, it depends on the sample. The transition from the feasible to the unfeasible region was found, at least for the Minimax and ES measures, to be smooth for finite $N,T$ and became sharp, with a well-defined critical value of $N/T$ separating the two regions, only for $N,T$ going to infinity with a fixed ratio between them \cite{kondor1}. This transition was identified as an algorithmic phase transition in \cite{kondor1}. (On algorithmic phase transitions see \cite{mezard}.) It may be conjectured that the instability for other coherent measures is similar, although at the present time nothing is known about their critical $N/T$ or the behavior around it.

Paradoxically, this instability is related to a very attractive feature of coherent measures: if one of the assets dominates the rest for all times (that is for infinitely large samples), then the coherent measures signal this by going to minus infinity. (This is a simple corollary of our Theorem \ref{theorem:crm} above). It may happen, however, that in a given finite sample a single asset dominates the others even if there is no such dominance relationship between them for infinitely long observation periods, and the coherent measures become unbounded from below also in this case, thereby giving a false signal. The probability of such a false alarm is known for any finite $N$ and $T$ for the Minimax problem \cite{kondor1}, and its behavior around the critical $N/T$ ratio is known for the Expected Shortfall \cite{ciliberti}. These results suggest that the probability is falling off rather fast as we go below the critical value of $N/T$, but in the context of other coherent measures this point would deserve further investigation.

In the somewhat extreme case of the Minimax problem we have been able to show via Theorem \ref{theorem:minimax} that the dominance of a single portfolio is not only sufficient but also necessary for the unfeasibility of the optimization. Theorem \ref{theorem:crm} states that the dominance of one portfolio is sufficient for the instability for any coherent measure, but the necessary conditions are not known even for ES. It is evident that they may be more relaxed than for the Minimax: the instability may set in even if none of the portfolios is dominant, but some of them perform exceptionally well on the given sample.

Although we have made use of all four coherence axioms to prove Theorem \ref{theorem:crm}, it is not clear whether this phenomenon is restricted to the coherent measures. On the other hand, it is obvious that there are risk measures (e.g. the variance or MAD) under which the optimization can always be carried out, provided the sample is large enough. It would be interesting to know the precise conditions that prevent this instability from setting in.

For the sake of simplicity, we have not put any restriction on short positions in this paper. Infinitely large long and short positions, or infinitely large values of the objective function can, of course, only occur if there is no constraint on short selling. As we briefly argued already above, however, possible constraints on the weights can only mask, rather than remove the instability. If we introduce constraints on the weights, they will make the domain within which the optimum is sought finite. It may still happen that one of the assets dominates the rest in a given sample, and then the objective function will monotonically decrease within the allowed domain along some direction. The solution will then be found where this unstable direction runs into the "wall" constituted by the constraint on the weights. For another sample, another unstable direction develops, so ultimately the solutions will not reflect any objective, stable feature of the risk measure, only the nature of random fluctuations. In a way, this is even worse than the appearance of runaway solutions in the unrestricted case: the fact that the solution sticks to the boundary of the allowed domain is much less conspicuous than the divergence of the risk measure, so it may well go unnoticed.

The introduction of the coherent risk axioms was a great achievement. They grasp some of the most important features we expect from any reasonable risk measure. We have seen, however, that coherent measures are sensitive to sample to sample fluctuations, and break down if one of the assets happens to dominate the others in a given sample. We argue that, in addition to mathematical consistency, noise tolerance, or robustness to sample to sample fluctuations under real life conditions where samples are always too small, is also a highly desirable feature of risk measures. It remains to be seen how much modification of the coherence framework is needed in order to accommodate the requirement of robustness.

\section*{Acknowledgements}
One of us (I. K.) is grateful to E. Berlinger for a short but inspiring discussion. This work has been supported by
the ''Cooperative Center for Communication Networks Data Analysis'', a NAP project
sponsored by the National Office of Research and Technology under grant No.\ KCKHA005.

\bibliography{feasibility}
\bibliographystyle{unsrt}

%\newpage
%\mbox{}

\end{document}